\documentclass[12pt]{article}
\usepackage{a4,amsmath,amsfonts,amsthm,latexsym,amssymb,graphicx}
\usepackage{fullpage}
\usepackage{tocloft}
\usepackage{bm}
\usepackage{enumerate}
\usepackage{enumitem}
\usepackage{bbm}
\usepackage{tikz}
\usepackage{tikz-cd}
\usepackage{ifthen}

\usepackage{caption}
\usepackage{subcaption}

\usepackage[mode=buildnew]{standalone}

\usepackage{hyperref}

\DeclareMathSymbol{\Q}{\mathalpha}{AMSb}{"51}
\DeclareMathSymbol{\R}{\mathalpha}{AMSb}{"52}
\DeclareMathSymbol{\Z}{\mathalpha}{AMSb}{"5A}
\DeclareMathSymbol{\N}{\mathalpha}{AMSb}{"4E}
\DeclareMathSymbol{\C}{\mathalpha}{AMSb}{"43}

\newcommand{\G}{\ensuremath{{\mathbb{G}}}}

\newcommand{\supp}{\ensuremath{\text{\rm supp}}}

\newcommand{\X}{\ensuremath{{\mathcal{X}}}}

\newcommand{\A}{\ensuremath{\mathcal{A}}}

\newcommand{\Patt}[2]{\ensuremath{\mathcal{L}_{#2}(#1)}}
\newcommand{\Lang}[1]{\ensuremath{\mathcal{L}(#1)}}

\renewcommand{\vec}[1]{\mathbf{#1}}

\newcommand{\bz}{\mathbf{z}}
\newcommand{\bu}{\mathbf{u}}
\newcommand{\bv}{\mathbf{v}}
\newcommand{\bt}{\mathbf{t}}

\newcommand{\bw}{\mathbf{w}}

\newcommand{\bb}{\mathbf{b}}

\newcommand{\bo}{\vec{0}}

\newtheorem{theorem}{Theorem}
\newtheorem{lemma}[theorem]{Lemma}
\newtheorem{corollary}[theorem]{Corollary}

\newtheorem*{claim*}{Claim}

\theoremstyle{definition}

\newtheorem{example}[theorem]{Example}
\newtheorem{remark}[theorem]{Remark}

\title{On the periodic decompositions of multidimensional configurations
}
\author{Pyry Herva and Jarkko Kari}
\date{Department of Mathematics and Statistics, University of Turku, Finland}
\begin{document}

\maketitle

\begin{abstract}
\noindent
% We consider $d$-dimensional configurations, that is, colorings of the $d$-dimensional integer grid $\Z^d$ with finitely many colors.
% Moreover, we interpret the colors as integers.
% It is known that if a configuration has a non-trivial annihilator, that is, if some non-trivial linear combination of its translations is the zero function, then it is a sum of finitely many periodic functions.
% This result is known as the periodic decomposition theorem.
% We prove two different improvements of the periodic decomposition theorem.
% The first improvement concerns configurations that have annihilators of specific type.
% The second improvement concerns configurations of particular type with non-trivial annihilators.
We consider $d$-dimensional configurations, that is, colorings of the $d$-dimensional integer grid $\Z^d$ with finitely many colors.
Moreover, we interpret the colors as integers so that configurations are functions $\Z^d \rightarrow \Z$ of finite range. We say that such function is $k$-periodic if it is invariant under translations in $k$ linearly independent directions.
1-periodic functions are called periodic.
It is known that if a configuration has a non-trivial annihilator, that is, if some non-trivial linear combination of its translations is the zero function, then it is a sum of finitely many periodic functions.
This result is known as the periodic decomposition theorem.
We prove two different improvements of it and discuss some applications of these improvements.
The first improvement gives a characterization on annihilators of a configuration to guarantee the $k$-periodicity of the functions in its periodic decomposition --- for any $k$. 
The periodic decomposition theorem is then a special case of this result with $k=1$. 
We discuss an application of this result concerning translational tilings.
The second improvement concerns so called sparse configurations for which the number of non-zero values in patterns grows at most linearly with respect to the diameter of the pattern. 
We prove that a sparse configuration with a non-trivial annihilator is a sum of finitely many periodic fibers where a fiber means a function whose support (that is, the set of points where the function gets non-zero values) is contained in a unique line.
As an application of this result, we show that $\R$-configurations with uniformly discrete supports that have non-trivial annihilators are necessarily periodic. 
\end{abstract}

\section{Introduction}

We work mostly with $d$-dimensional configurations $c\in \A^{\Z^d}$, that is, colorings of the $d$-dimensional integer grid $\Z^d$ with a non-empty finite set of colors $\A$ --- where $d$ is any positive integer.
In this paper we assume that $\A \subseteq \Z$.
% We study configurations that exhibit a form of weak periodicity.
We consider configurations $c$ that have non-trivial annihilators meaning that a non-trivial $\Z$-linear combination of some finitely many translated copies of $c$ is the zero function.
It is known that such configurations can be decomposed to sums of finitely many periodic functions \cite{fullproofs} --- we will refer to this result as the periodic decomposition theorem (Theorem \ref{thm: decomposition theorem}).
In this paper we show that under some additional assumptions the periodic decomposition theorem can be improved.
We prove two different improvements of the periodic decomposition theorem.

% MOTIVOINTIA VÄHÄN ENEMMÄN 
% NIVAT
% PERIODIC TILING CONJECTURE
% YMS...

In our considerations we apply an algebraic approach to multidimensional symbolic dynamics which was developed in \cite{icalp, fullproofs} to make progress in the famous still open Nivat's conjecture \cite{Nivat} which links periodicity and sufficiently small rectangular pattern complexity.
% More precisely, Nivat's conjecture states that
It is a two-dimensional generalization of the Morse-Hedlund theorem \cite{morse-hedlund}.
In the algebraic approach $d$-dimensional configurations are represented as formal power series with integer coefficients in $d$ variables.
An annihilator of a configuration is then a Laurent polynomial such that its formal product with the power series presenting the configuration is the zero power series.
A periodizer of a configuration is a Laurent polynomial such that its formal product with
the power series presenting the configuration is strongly periodic, that is, it has $d$ linearly independent period vectors.

In our first improvement of the periodic decomposition theorem we consider arbitrary configurations that have periodizers of particular type.
% This result is very much inspired by similar results in \cite{meyerovitch2023}.
More precisely, we assume that for some $k \in \{1,\ldots,d\}$ and for every $(k-1)$-dimensional linear subspace $V$ of $\R^d$ the configuration has a periodizer whose support (that is, the set that gives the non-zero terms of the periodizer) contains exactly one point of $V$.
We show that the configuration is then a sum of finitely many functions which all have $k$ linearly independent periods (Theorem \ref{thm: main result 1}).
If $k=1$, then the statement says just that any configuration with a non-trivial periodizer is a sum of finitely many periodic functions which is exactly the original periodic decomposition theorem.
Also, for $k=d$ the statement is known to be true \cite{DLT_invited_jarkko}.

The result described above is inspired by similar considerations in \cite{meyerovitch2023} by Meyerovitch et al. related to translational tilings and the periodic tiling conjecture.
A translational tiling by a non-empty finite set $D \subseteq \Z^d$ is a set $C \subseteq \Z^d$ such that $D \oplus C=\Z^d$. The non-empty finite set $D$ is called a tile, and the tiling $C$ may also be called a co-tiler of $D$.
Considering the indicator function of the co-tiler, we find that it is periodized by the polynomial presenting the tile $D$.
The periodic tiling conjecture states that if a tile has a co-tiler, then it also has a periodic co-tiler \cite{grunbaum-shepard,Lagarias-wang, stein}.
Newman showed already in 1977 using a simple pigeonholing argument that any translational tiling in dimension $d=1$ is necessarily periodic \cite{tesselation} and hence the periodic tiling problem is true for $d=1$.
For $d \geq 2$, the conjecture is much trickier.
% Only very recently, 
However, Bhattacharya proved 
% in a paper \cite{bhattacharya} published in 2020 
that the periodic tiling problem holds also for $d=2$ \cite{bhattacharya}.
% In \cite{Greenfeld2021structure}
Greenfeld and Tao proved a quantative version of Bhattacharyas result \cite{Greenfeld2021structure}.
Moreover, they obtained some structural results on translational tilings in any dimension.
In \cite{greenfeld-tao2022} they showed that the periodic tiling problem fails for sufficiently large $d$.
% In \cite{meyerovitch2023} the authors consider joint co-tilers.
For recent research around the periodic tiling problem, see \emph{e.g.}
\cite{grebik-greenfeld-rozhon-tao, greenfeld2023undecidabilitytranslationalmonotilings, greenfeld-tao-undecidable, meyerovitch2022notereductiontilingproblems, meyerovitch2023, greenfeld2023tilingspectralityaperiodicityconnected, pont2024periodicitydecidabilitytranslationaltilings}.

% As an application of the improvement we consider translational tilings..

% As a corollary (Corollary \ref{corollary: independent periodizers}) of this result we get a similar result as in \cite{meyerovitch2023} concerning translational tilings.
% This corollary is closely related
% to the still partially open periodic tiling conjecture \cite{Lagarias-wang,stein}.
% The periodic tiling conjecture is true in dimensions $d=1$ \cite{tesselation} and $d=2$ \cite{bhattacharya} but fails for sufficiently large $d$ \cite{greenfeld-tao2022}.

% There is an application of this result concerning translational tilings (Corollary \ref{corollary: independent periodizers}).
% This application is closely related to a similar result in \cite{meyerovitch2023} and to the still partially open periodic tiling conjecture \cite{Lagarias-wang,stein}.
% The periodic tiling conjecture is true in dimensions $d=1$ \cite{tesselation} and $d=2$ \cite{bhattacharya} but fails for sufficiently large $d$ \cite{greenfeld-tao2022}.

In our second improvement of the periodic decomposition theorem we consider configurations whose supports (that is, the sets of points where the configurations get non-zero values) are located very ``sparsely'' in the grid $\Z^d$.
Conveniently, we call these configurations sparse.
We show that if a sparse configuration has a non-zero annihilator, then it is a sum of finitely many periodic configurations whose supports are aligned on unique lines (Theorem \ref{thm: periodic decomposition sparse} and Corollary \ref{corollary: sparse configuration with annihilators is a sum of finitely many periodic fibers}).

The motivation to study sparse configurations with annihilators arises from the authors' study \cite{delone} concerning Delone sets where the concept of annihilators is generalized to $\R^d$-configurations, that is, functions $\R^d \rightarrow \Z$ of finite range.
Delone sets are uniformly discrete and relatively dense subsets of $\R^d$ that form mathematical models for crystals and quasicrystals \cite{aperiodic1}.
By the support of an $\R^d$-configuration we mean the set of points where the configuration gets non-zero values.
Thus, in \cite{delone} we study $\R^d$-configurations whose supports are Delone sets.
In \cite{delone}, many results concerning $\Z^d$-configurations are generalized to $\R^d$-configurations.
Our research question in \cite{delone} related to this article originally asked whether one-dimensional $\R$-configurations with Delone set supports that have non-trivial annihilators are necessarily periodic.
Using our periodic decomposition theorem for sparse configurations we prove that this holds even for $\R$-configurations with uniformly discrete supports that have non-trivial annihilators (Theorem \ref{thm: delone application}).
% $\R$-configurations with Delone set supports are studied...
% VIITTEITÄ DELONE-JOUKKOIHIN LIITTYEN

% This version of the periodic decomposition theorem is used in the study of colorings of 1-dimensional Delone sets with non-trivial annihilators and in proving that this kind of colorings are necessarily periodic \cite{delone}.

This article is an extended version of the conference paper \cite{sofsem2025}.

\subsection*{Structure of the paper}

We begin with preliminaries in Section \ref{sec: preliminaries} by presenting the necessary concepts and definitions and some past results relevant to us.
In Sections \ref{sec: main result 1} and \ref{sec: main result 2} we state and prove our improvements of the periodic decomposition theorem and discuss some related results.
The main result of Section \ref{sec: main result 1}, that is, our first improvement of the periodic decomposition theorem is Theorem \ref{thm: main result 1} and its application to translational tilings is Corollary \ref{corollary: independent tiles}.
The main result of Section \ref{sec: main result 2}, that is, our periodic decomposition theorem for sparse configurations is Theorem \ref{thm: periodic decomposition sparse}
and the application of this result is Theorem \ref{thm: delone application}.
% We conclude with Section \ref{sec: conclusions}.

\section{Preliminaries} \label{sec: preliminaries}

% We begin by presenting the necessary concepts and definitions.

\subsection{Notation}

As usual, we denote by $\Z$, $\Q$, $\R$, and $\C$ the sets of integers, rational numbers, real numbers, and complex numbers, respectively.
By $\Z_+$ we mean the set of positive integers and by $\N= \Z_+ \cup \{0\}$ the set of non-negative integers.
% We denote by $\lfloor x \rfloor$ and $\{x\}= x - \lfloor x \rfloor$ the integer and fractional parts of a real number $x \in \R$, respectively.
% We use common notations of linear algebra.
For any two sets $A$ and $B$, we denote by $A^B$ the set of all functions from $B$ to $A$.
For any vectors $\bv_1,\ldots,\bv_k \in \R^d$, we denote
by $L(\bv_1,\ldots,\bv_k) = \{ a_1\bv_1 + \ldots+a_k \bv_k \mid a_1,\ldots,a_k \in \R \}$
the linear subspace of $\R^d$ spanned by them.
Similarly for any (possibly infinite) set $S\subseteq \R^d$, we denote by $L(S)$ the subspace generated by $S$.
On the other hand, by $\langle \bv_1,\ldots,\bv_k \rangle$ we mean the subgroup of $\R^d$ generated by the vectors $\bv_1,\ldots,\bv_k$.
We also use the notation $\Z[S]$ to denote the abelian group generated by the set $S$.
By $\G_k$ we denote the set of all $k$-dimensional subspaces of $\R^d$ for a fixed $d \in \Z_+$.
Clearly, if $k>d$, then $\G_k = \emptyset$.

\subsection{Configurations and periodicity}

In this paper we consider mostly elements of the set $\Z^{\Z^d}$, that is, functions 
$$
c \colon \Z^d \rightarrow \Z
$$ 
where $d$ is a positive integer --- the dimension.
However, in some occasions, we also work with the set $\C^{\Z^d}$.
% In the end of the paper, we also shortly consider the set $\C^{\R}$.
% In the or the set $\C^{\R}$.
and hence most of the definitions and some of the statements are given more generally for the elements of this set.
% Also, functions of the sets $\Q^{\Z^d}$, $\R^{\Z^d}$, or $\C^{\Z^d}$ may be considered.

% \begin{definition}
A ($d$-dimensional) \emph{configuration} or a \emph{$\Z^d$-configuration} is a function
$c \in \Z^{\Z^d}$ that has only finitely many different values.
% then $c$ is called a ($d$-dimensional) \emph{configuration}.
% \end{definition}
% The support of $c$ is the set $\supp(c) = \{\bu \mid c(\bu) \neq 0\}$.
% \noindent
In particular, configurations are functions of the set $\A^{\Z^d}$ for some finite $\A \subseteq \Z$.
% which is called the \emph{alphabet} of $c$.
The set $\A^{\Z^d}$ is called the \emph{$d$-dimensional configuration space over $\A$}.
Notice the distinction between the concepts of configurations and general functions of the set $\Z^{\Z^d}$. Indeed, the range of a general such function is allowed to be infinite while the range of a configuration is always finite.

% The translation $\tau^{\bv}(c)$ of a configuration $c$ by a vector $\bv \in \Z^d$ is defined such that $\tau^{\bv}(c)(\bu)$

Any function $c \in \C^{\Z^d}$ is \emph{$\bv$-periodic} if $c(\bu + \bv) = c(\bu)$ for all $\bu \in \Z^d$. In this case $\bv$ is called a \emph{period} of $c$.
If $c$ has a non-zero period, then it is \emph{periodic}.
% Otherwise, it is called non-periodic.
Let us denote by $\tau^{\bv}(c)$ the \emph{translation} of $c \in \C^{\Z^d}$ by a vector $\bv \in \Z^d$ defined such that $\tau^{\bv}(c)(\bu) = c( \bu -\bv)$ for all $\bu \in \Z^d$.
Hence, $c$ is $\bv$-periodic if and only if $\tau^{\bv}(c) =c$.

A function $c \in \C^{\Z^d}$ is \emph{$k$-periodic} if 
it has $k$ linearly independent periods.
In particular, $d$-periodic functions are called \emph{strongly periodic}.
(Note that if $c \in \Z^{\Z^d}$ is strongly periodic, then $c$ is necessarily a configuration.)
We say that
$c \in \C^{\Z^d}$ is \emph{periodic in direction $\bv \in \Q^d$} if it is $k \bv$-periodic for some $k \in \Z$. 
It is \emph{$V$-periodic} for a vector space $V \subseteq \R^d$ if it is periodic in direction $\bv$ for all $\bv \in V \cap \Q^d$.
Clearly, any $k$-periodic function is $V$-periodic for some $V \in \G_k$.
Finally, $c \in \C^{\Z^d}$ is \emph{$\bv$-periodic in a subset} $S \subseteq \Z^d$ if for all $\bu \in S$ such that $\bu + \bv \in S$ it follows that $c(\bu) = c(\bu+\bv)$.

A \emph{(finite) pattern} is a function $p \in \C^D$ for some finite \emph{shape} $D \subseteq \Z^d$.
In this case $p$ is also called a \emph{$D$-pattern}.
We say that a function $c \in \C^{\Z^d}$ \emph{contains} a $D$-pattern $p \in \C^D$ if $\tau^{\bt}(c) \restriction _D = p$ for some $\bt \in \Z^d$.
We denote by $\Patt{c}{D} = \{ \tau^{\bt}(c) \restriction_D \mid \bt \in \Z^d\}$ the set of of all $D$-patterns of $c$ and by 
$$
\Lang{c} = \bigcup_{\text{finite } 
 D\subseteq\Z^d} \Patt{c}{D}
$$ 
the set of all finite patterns of $c$.

The \emph{support} of $c \in \C^{\Z^d}$ 
% (or any function $c\in \C^G$ where $G$ is any group)
is the set 
$$
\supp(c) = \{\bu \mid c(\bu) \neq 0\}.
$$
If $\supp(c)$ is finite, then $c$ is called \emph{finitely supported}.

Let us denote by $C_m = \{ -m , \ldots , m\}^d$ for $m \in \N$ the discrete $d$-dimensional hypercube of size $(2m+1)^d$ centered at the origin. 
% In this paper 
Any function $c \in \C^{\Z^d}$ is called \emph{sparse} if there exists a positive integer $a$ such that
$$
|\supp(c) \cap (C_m + \bt)| \leq am
$$
for all $m \in \Z_+$ and for all $\bt \in \Z^d$.
Such $a$ is called a sparseness constant of $c$.

% Let us review some basic concepts of symbolic dynamics we need. For a reference see \emph{e.g.} \cite{tullio,kurka,lindmarcus}.
The configuration space $\A^{\Z^d}$ can be made a compact topological space by endowing $\A$ with the discrete topology and considering the product topology it induces on $\A^{\Z^d}$ -- the \emph{prodiscrete topology}.
This topology is
induced by a metric where two configurations are close if they agree on a large area around the origin.
By compactness in this topology every sequence of configurations has a converging subsequence.

A \emph{subshift} is a closed and translation-invariant subset of the configuration space $\A^{\Z^d}$.
% Equivalently, subshifts can be defined by using forbidden patterns:
% Given a set $F \subseteq \A^*$ of \emph{forbidden} patterns, the set
% $$
% X_F=\{c \in \A^{\Z^d} \mid  \Lang{c} \cap F = \emptyset \}
% $$
% of configurations that avoid all forbidden patterns
% is a subshift. Moreover, every subshift is obtained by forbidding some set of finite patterns.
% If $F \subseteq \A^*$ is finite, then $X_F$ is a \emph{subshift of finite type} (SFT).
The \emph{orbit} of a configuration $c \in \A^{\Z^d}$ is the set $\mathcal{O}(c) = \{ \tau^{\bt}(c) \mid \bt \in \Z^d \}$ of its every translate.
The
\emph{orbit closure} $\overline{\mathcal{O}(c)}$ is the topological closure of its orbit under the prodiscrete topology.
It is the smallest subshift that contains $c$.
It consists of all configurations $c'$ such that $\Lang{c'}\subseteq \Lang{c}$.

\subsection{Algebraic concepts}

As in \cite{icalp} we present a function $c \in \C^{\Z^d}$ as a formal power series 
$$
c(X) = \sum_{\bu \in \Z^d} c_{\bu} X^{\bu}
$$
in $d$ variables $X=(x_1,\ldots,x_d)$
where we have used the convenient abbreviations $X^{\bu} = x_1^{u_1}\cdots x_d^{u_d}$ and $c_{\bu} = c(\bu)$ for a vector $\bu=(u_1,\ldots,u_d) \in \Z^d$.

We consider Laurent polynomials with integer coefficients which we call from now on simply polynomials.
By a non-trivial polynomial we mean a non-zero polynomial.
A multiplication of a formal power series $c(X)=\sum_{\bu\in\Z^d}c_{\bu} X^{\bu}$ with complex coefficients by a polynomial $f=f(X)=\sum_{i=1}^nf_i X^{\bu_i}$ is well defined as the formal power series
$$
fc=f(X) c(X) =
\sum_{\bu \in \Z^d} \left (\sum_{i=1}^n f_i c_{\bu} \right ) X^{\bu+\bu_i}
=\sum_{\bu \in \Z^d} \left (\sum_{i=1}^n f_i c_{\bu-\bu_i} \right ) X^{\bu}.
$$
Thus, the value of $fc$ at $\bu \in \Z^d$ is determined by the values of $c$ at $\bu-\bu_1, \ldots , \bu-\bu_n$.

\begin{remark}
More generally, we could and sometimes need to consider also Laurent polynomials with complex coefficients but in this paper polynomials are Laurent polynomials with integer coefficients as defined above.
\end{remark}

By the product of a function $c \in \C^{\Z^d}$ and a polynomial $f$ we mean the above product and denote it conveniently by $fc$ as usual.
Since any polynomial $f$ can be thought as a finitely supported function in $\Z^{\Z^d}$, the above product can also be viewed as the discrete convolution of $f$ and $c$ (which is always defined because $f$ is finitely supported).
By the support of a polynomial $f = \sum f_{\bu} X^{\bu}$ we mean the (finite) support of this function, that is, the set
$$
\supp(f) = \{ \bu \mid f_{\bu} \neq 0 \}.
$$

We say that a polynomial
$f$
\emph{annihilates} (or is an \emph{annihilator} of) a function $c \in \C^{\Z^d}$ if
$
fc = 0.
$
Since multiplying $c \in \C^{\Z^d}$ by the monomial $X^{\bv}$ corresponds to translating it by the vector $\bv$, that is, $X^{\bv}c=\tau^{\bv}(c)$, it follows that $c$ is $\bv$-periodic if and only if it is annihilated by the difference polynomial $X^{\bv}-1$.
% Hence, annihilation by a non-trivial polynomial can be seen as a form of weak periodicity.
The set of every annihilator of $c$ is quite easily seen to be an ideal of the polynomial ring
and hence the question whether $c$ is periodic reduces to the question whether its ``annihilator ideal'' contains a non-trivial difference polynomial.

We say that a polynomial $f$ \emph{periodizes} (or is a \emph{periodizer} of) $c \in \C^{\Z^d}$ if $fc$ is strongly periodic.
Clearly, any annihilator of $c$ is also a periodizer of $c$.
On the other hand, if $c$ has a periodizer $f$, then $fc$ is periodic and hence annihilated by some non-trivial difference polynomial $X^{\bv}-1$.
Thus, $c$ has a non-trivial annihilator if and only it has a non-trivial periodizer.

A \emph{line polynomial} $f$ is a polynomial whose support contains at least two points and the points of the support are aligned on a line, that is, there exist vectors $\bu,\bv \in \Z^d$ such that $\supp(f) \subseteq \bu + \Q \bv$.
Such $\bv$ (which is non-zero) is called a \emph{direction} of $f$ and we may say that $f$ is a \emph{line polynomial in direction $\bv$}.
As usual, we say that two vectors $\bv,\bv'$ are \emph{parallel} if $\bv' \in \Q \bv$ (and vice versa), {\it i.e.}, if they are linearly dependent over $\Q$.
Otherwise, they are \emph{non-parallel} in which case they are linearly independent over $\Q$.
Clearly, any line polynomial in direction $\bv$ is also a line polynomial in any parallel non-zero direction $\bv'$.

% FIBERS
A function $c \in \C^{\Z^d}$ is a \emph{$\bv$-fiber} for a non-zero vector $\bv \in \Q^d$ if its support is contained in a line in direction $\bv$, that is, if $\supp(c) \subseteq \bu + \Q \bv$ for some $\bu \in \Z^d$.
Any $\bv$-fiber is called simply a \emph{fiber}.
Note that by interpreting polynomials as functions as discussed earlier, line polynomials are finitely supported fibers.
% Hence, fibers generalize the concept of line polynomials.
By a \emph{$\bv$-fiber of $c \in \C^{\Z^d}$} we mean a function that agrees with $c$ on $\bu + \Q \bv$ for some $\bu \in \Z^d$ and gets value 0 elsewhere.
Clearly, any $\bv$-fiber of $c$ is a $\bv$-fiber.

Note that if a configuration $c \in \A^{\Z^d}$ is annihilated by a line polynomial, then it is necessarily periodic.
Indeed, annihilation by a line polynomial in direction $\bv$ defines a recurrence relation on each $\bv$-fiber of $c$ which implies periodicity of each $\bv$-fiber since $\A$ is finite.
Moreover, the smallest periods of the $\bv$-fibers are bounded and hence also $c$ is $\bv$-periodic.
In particular, any one-dimensional configuration with a non-trivial annihilator is periodic since all polynomials in only one variable are line polynomials.
Conversely, periodicity implies annihilation by a non-trivial difference polynomial which is a line polynomial.
Thus, a configuration $c$ is periodic if and only if it is annihilated by a line polynomial.

\subsubsection*{Special annihilators}

It is known that if a configuration $c \in \A^{\Z^d}$ has a non-trivial periodizer, then it has an annihilator which is a product of difference polynomials.
% To state this result, precisely we need some terminology.
% A \emph{dilation} of a polynomial $f(X)$ is a polynomial of the form $f(X^k)$ for some integer $k$.
More precisely, we have the following theorem.
% By a dilation con

\begin{theorem}[\cite{icalp}, \cite{fullproofs}] \label{thm: special annihilator}
Let $c\in \A^{\Z^d}$ be a configuration and assume that it has a non-trivial annihilator $f$.
For all $\bu \in \supp(f)$, there exist pairwise non-parallel vectors $\bv_1,\ldots,\bv_m \in \Z^d$ such that each $\bv_i$ is an integer multiple of $\bu_i-\bu$ for some $\bu_i \in \supp(f) \setminus \{\bu\}$, and $c$ is annihilated by the polynomial
$$
(X^{\bv_1}-1) \cdots (X^{\bv_m}-1).
$$
\end{theorem}

For some considerations of this paper, we need
a more involved formulation of the above theorem.
To do this, we need to review some parts of the proof of the theorem.
A \emph{dilation} of a polynomial $f(X)$ is a polynomial of the form $f(X^k)$ for some integer $k$.

\begin{lemma}[The dilation lemma \cite{fullproofs}]
  Let $c \in \A^{\Z^d}$ be a configuration and $f(X)$ a non-trivial annihilator of $c$. 
  There exists a positive integer $r$ such that for every positive integer $k$ with $\gcd(k,r)=1$ also $f(X^k)$ annihilates $c$.
\end{lemma}

\noindent
Let us call a number $r$ that has the property of the above lemma a \emph{dilation constant} of $c$ with respect to $f$.
The following lemma was proved in the proof of Theorem \ref{thm: special annihilator}.

\begin{lemma}[\cite{fullproofs}]
% \label{thm: special annihilator}
\label{lemma1}
Let $c\in\A^{\Z^d}$ be a configuration and $f$ a non-trivial annihilator of $c$.
For all $\bu \in \supp(f)$, the configuration $c$ is annihilated by the polynomial
$$
\prod_{\bv \in \supp(f) \setminus \{\bu\}} (X^{r(\bv-\bu)}-1)
$$ 
where $r$ is any dilation constant of $c$ with respect to $f$.
% If $r$ is a dilation constant of $c$ with respect to $f$, then for all $\bu \in \supp(f)$ the configuration $c$ is annihilated by the polynomial
% $$
% \prod_{\bv \in \supp(f) \setminus \{\bu\}} (X^{r(\bv-\bu)}-1).
% $$ 
\end{lemma}

\noindent
From the proof of the dilation lemma we adapt the following result.
% For completeness, we provide a short description of the proof.

\begin{lemma}[Adapted from \cite{fullproofs}] \label{lemma2}
  Let $\mathcal{I}$ be an arbitrary index set, and let $\left ( c^{(i)} \right ) _{i \in \mathcal{I}}$ be a collection of configurations over the same finite alphabet $\A \subseteq \Z$.
  If $f$ is a non-trivial annihilator of $c^{(i)}$ for every $i \in \mathcal{I}$, then the configurations in the collection have a common dilation constant with respect to $f$.
\end{lemma}

\begin{proof}
Let $c_{\text{max}}$ be the maximum absolute value of the coefficients of the configurations $c^{(i)}$.
Since the configurations $c^{(i)}$ are over the same alphabet, such number exists.
Let $f = \sum_{\bv \in \supp(f)} f_{\bv}X^{\bv}$ and define $s= c_{\text{max}} \sum_{\bv \in \supp(f)} |f_{\bv}|$.
In the proof of the dilation lemma in \cite{fullproofs}, it was shown that $r=s!$ is a dilation constant of any $c^{(i)}$.
The claim follows.
\end{proof}

\noindent
We need one more lemma from \cite{fullproofs}.

\begin{lemma}[\cite{fullproofs}]
\label{lemma3}
Let $c$ be a configuration. If $\varphi_1,\ldots,\varphi_m$ are line polynomials such that $\varphi_1^{e_1}\cdots\varphi_m^{e_m}$ annihilates $c$ for some positive integers $e_1,\ldots,e_m$,
then also $\varphi_1\cdots\varphi_m$ annihilates $c$.
\end{lemma}

\noindent
Combining the above lemmas we have the following result.

\begin{theorem}
\label{theorem: special annihilator common}
Let $\mathcal{I}$ be an arbitrary index set, and let $\left ( c^{(i)} \right ) _{i \in \mathcal{I}}$ be a collection of configurations over the same finite alphabet $\A \subseteq \Z$.
  If $f$ is a non-trivial annihilator of $c^{(i)}$ for every $i \in \mathcal{I}$, then the configurations in the collection have a common annihilator which is of the form 
  $$
  (X^{\bv_1}-1)\cdots(X^{\bv_m}-1)
  $$
  where the vectors $\bv_1,\ldots,\bv_m$ are pairwise non-parallel.
\end{theorem}

\begin{proof}
By Lemma \ref{lemma2}, the configurations $c^{(i)}$ have a common dilation constant $r$ and hence by Lemma \ref{lemma1} they have a common annihilator
$$
g=\prod_{\bv \in \supp(f) \setminus \{\bu\}} (X^{r(\bv-\bu)}-1) =(X^{\bw_1}-1)\cdots (X^{\bw_k}-1).
$$
Assume that for some $i\neq j$, the vectors $\bw_i$ and $\bw_j$ are parallel, that is,
$q\bw_i= p \bw_j$ for some $p,q \in \Z\setminus\{0\}$.
  % Hence, $q \bv_i = p \bv_j$.
We may replace $X^{\bw_i}-1$ in $g$ by $X^{p \bw_j}-1 = X^{q \bw_i}-1$.
  Since 
  $$
  X^{p \bw_j}-1 = (X^{\bw_j}-1)(X^{(p-1)\bw_j} + \ldots + X^{\bv_j} + 1),
  $$
   we may replace $(X^{p\bw_j}-1)(X^{\bw_j}-1)$ by $(X^{p\bw_j}-1)^2$.
  Iterating this argument we conclude that the configurations $c^{(i)}$ have a common annihilator
  $$
  (X^{\bv_1}-1)^{e_1} \cdots (X^{\bv_m}-1)^{e_m}
  $$
  where the vectors $\bv_1,\ldots,\bv_m$ are pairwise linearly independent over $\Q$ and $e_1,\ldots,e_m$ are positive integers.
  Consequently, by Lemma \ref{lemma3}, we conclude that the configurations $c^{(i)}$ have a common annihilator
  $$
    (X^{\bv_1}-1)\cdots(X^{\bv_m}-1).
  $$
\end{proof}

\subsubsection*{Periodic decomposition theorem}

The multiplication of $c$ by a difference polynomial can be thought of as a ``discrete derivation'' of $c$.
Theorem \ref{thm: special annihilator} says that if a configuration $c\in \A^{\Z^d}$ has a non-trivial periodizer, then there is a sequence of derivations which annihilates $c$.
So, by
``integrating'' step by step, we have the following periodic decomposition theorem.

\begin{theorem}[Periodic decomposition theorem \cite{fullproofs}] \label{thm: decomposition theorem}
Let $c\in \A^{\Z^d}$ be a configuration and assume that it has a non-trivial periodizer.
There exist pairwise non-parallel vectors $\bv_1,\ldots,\bv_m$ and functions $c_1,\ldots,c_m \in \Z^{\Z^d}$ such that
$$
c = c_1 + \ldots + c_m
$$
and each $c_i$ is $\bv_i$-periodic.
\end{theorem}

\begin{remark}\label{remark: periodic decomposition theorem}
The periodic functions $c_i$ in the periodic decomposition of a configuration $c$ may not be configurations, that is, their ranges may be infinite \cite{kari-szabados-multidimensional_words}.
\end{remark}

\section{Periodic decomposition of configurations with annihilators of particular type} \label{sec: main result 1}

In this section we prove our first improvement of the periodic decomposition theorem.
For the proof, we need some lemmas.

\begin{lemma} \label{lemma: auxiliary lemma 1 for main result 1}
Let $V\subseteq \R^d$ be a linear subspace of $\R^d$,
and let $\varphi$ and $\psi$ be line polynomials in non-parallel directions $\bv_1 \in \Z^d$ and $\bv_2 \in \Z^d$, respectively.
Also, assume that 
$$
\langle \bv_1,\bv_2 \rangle \cap V = \{ \bo \}.
$$
% $\bv_1$ is not in the subspace generated by $\bv_2$ and $V$, and
% $\bv_2$ is not in the subspace generated by $\bv_1$ and $V$.
If $c' \in \C^{\Z^d}$ is a $V$-periodic function annihilated by $\psi$,
then there exists $c \in \C^{\Z^d}$ such that the following conditions hold: 
\begin{itemize}
    \item $\varphi c = c'$,
    \item $\psi c =0$, and
    \item $c$ is $V$-periodic.
\end{itemize}
Moreover, if $c' \in \Z^{\Z^d}$ and $\varphi$ is a difference polynomial, then the coefficients of $c$ can be chosen to be integers, that is, $c \in \Z^{\Z^d}$.
\end{lemma}

\begin{proof}
% Let $\bv_1$ and $\bv_2$ be directions of $\varphi$ and $\psi$, respectively.
Without loss of generality we may assume that $\varphi = \alpha_0 + \alpha_1 X^{\bv_1}+\ldots+\alpha_nX^{n\bv_1}$ for some positive integer $n$ such that $\alpha_0,\alpha_n \neq 0$.
Let $\{\bb_1, \ldots, \bb_k\} \subseteq \Z^d$ be a base of $V\cap\Q^d$.

The space $\Z^d$ is partitioned into cosets modulo 
$$
% \Z[\bv_1,\bv_2,\bb_1,\ldots,\bb_k] 
\langle \bv_1,\bv_2, \bb_1,\ldots,\bb_k \rangle
= 
\{a_1 \bv_1 + a_2\bv_2+b_1\bb_1+\ldots+b_k\bb_k \mid a_1,a_2,b_1,\ldots,b_k \in \Z \}.
$$
Let us fix a point $\bz_{\Lambda} \in \Lambda$ for each such coset $\Lambda$.
Now, we have
% Let us fix an arbitrary $\bz \in \Z^d$ and consider the coset
$\Lambda = \bz_{\Lambda} + \langle \bv_1,\bv_2, \bb_1,\ldots,\bb_k \rangle$.
Let us denote 
$$
\Lambda[a_1,a_2,b_1,\ldots,b_k] = \bz_{\Lambda} +  a_1 \bv_1 + a_2\bv_2+b_1\bb_1+\ldots+b_k\bb_k.
$$
Since by the assumption
$\langle \bv_1,\bv_2 \rangle \cap V = \{ \bo \}$,
each element of $\Lambda$ has a unique expression as $\Lambda[a_1,a_2,b_1,\ldots,b_k]$.
In other words, if $$\Lambda[a_1,a_2,b_1,\ldots,b_k] = \Lambda[a_1',a_2',b_1',\ldots,b_k'],$$ then
$a_1=a_1',a_2=a_2',b_1=b_1',\ldots,b_k=b_k'$.

The equation $\varphi c=c'$ is satisfied in an arbitrary coset $\Lambda$ if and only if
\begin{equation}\label{equation:recurrence}
\alpha_n c_{\Lambda[a_1-n,a_2,b_1,\ldots,b_k]} + \ldots 
% + \alpha_1 c_{\Lambda[a_1-1,a_2,b_1,\ldots,b_k]} 
+ 
\alpha_0 c_{\Lambda[a_1,a_2,b_1,\ldots,b_k]} = c'_{\Lambda[a_1,a_2,b_1,\ldots,b_k]} 
\end{equation}
for all $a_1,a_2,b_1,\ldots,b_k \in \Z$.
% and for every coset $\Lambda_{\bz}$.
Let us define $c_{\Lambda[a_1,a_2,b_1,\ldots,b_k]} = 0$ for $a_1 \in [0,n)$.
The rest of $c$ is then determined by the linear recurrence of Equation \eqref{equation:recurrence} such that $\varphi c = c'$.

Let us show that if $f$ is a line polynomial in direction $\bw \in \{\bv_2,\bb_1,\ldots, \bb_k\}$ and annihilates $c'$, then it also annihilates $c$:
Since $fc'=0$, we have
\begin{equation}\label{eq:recurrence2}
\varphi (fc) = f( \varphi c) = fc' = 0.
\end{equation}
We defined $c_{\Lambda[a_1,a_2,b_1,\ldots,b_k]} = 0$ for all $a_1 \in [0,n)$ and $a_2,b_1,\ldots,b_k \in \Z$. Thus, it follows that also
$(f c)_{\Lambda[a_1,a_2,b_1,\ldots,b_k]} = 0$ for all $a_1 \in [0,n)$ and $a_2,b_1,\ldots,b_k \in \Z$ and hence the linear recurrence \eqref{eq:recurrence2} implies that $(f c)_{\Lambda[a_1,a_2,b_1,\ldots,b_k]}=0$ for all $a_1,a_2,b_1,\ldots,b_k \in \Z$.
This holds in every coset $\Lambda$ and hence $f c = 0$.

% Since $\psi c' = 0$, we have
% $$
% \varphi (\psi c) = \psi (\varphi c) = \psi c'= 0,
% $$
% % It follows that $\psi c =0$. 
% and since we defined $c_{\Lambda[a_1,a_2,b_1,\ldots,b_k]} = 0$ for all $a_1 \in [0,n)$ and $a_2,b_1,\ldots,b_k \in \Z$, it follows that also
% $(\psi c)_{\Lambda[a_1,a_2,b_1,\ldots,b_k]} = 0$ for all $a_1 \in [0,n)$ and $a_2,b_1,\ldots,b_k \in \Z$ and hence the above recurrence implies that $(\psi c)_{\Lambda[a_1,a_2,b_1,\ldots,b_k]}=0$ for all $a_1,a_2,b_1,\ldots,b_k \in \Z$.
% This holds in every coset $\Lambda$ and hence $\psi c = 0$.

% By a similar argument as above, since $c'$ is periodic in direction $\bb_i$ for each $i \in \{1,\ldots,k\}$, that is, annihilated by a difference polynomial in direction $\bb_i$, it follows that also $c$ is annihilated by this same difference polynomial and hence periodic in direction $\bb_i$.
% % Thus, $c$ is $V$-periodic.
% Thus, $c$ is periodic in direction $\bv$ for every $\bv \in V \cap \Q^d$ and hence by the definition $c$ is $V$-periodic.

Now, it follows that $\psi c = 0$ since $\psi$ annihilates $c'$ and has direction $\bv_2$.
Also, since
$c'$ is periodic in direction $\bb_i$ for each $i \in \{1,\ldots,k\}$, that is, annihilated by a difference polynomial in direction $\bb_i$, it follows that also $c$ is annihilated by this same difference polynomial and hence periodic in direction $\bb_i$.
Thus, $c$ is periodic in direction $\bv$ for every $\bv \in V \cap \Q^d$ and hence by the definition $c$ is $V$-periodic.

Finally, assume that $c' \in \Z^{\Z^d}$ and $\varphi$ is a difference polynomial, that is, $\varphi = X^{ t \bv_1} - 1$ for some $t \in \Z$.
Now, Equation \eqref{equation:recurrence} turns into form
$$
c_{\Lambda[a_1 - t,a_2,b_1 , \ldots , b_k]} - c_{\Lambda[a_1,a_2,b_1 , \ldots , b_k]} = c'_{\Lambda[a_1,a_2,b_1 , \ldots , b_k]}.
$$
% Since the values of $c'$ are integers, it follows that the values of $c$ can be chosen to be integers as well.
Since the values of $c'$ are integers, it follows that the values of $c$ as defined above are integers as well.
% \qed
\end{proof}

\begin{lemma} \label{lemma: auxiliary lemma 2 for main result 1}
Let $V\subseteq \R^d$ be a linear subspace of $\R^d$.
Assume that $c \in \C^{\Z^d}$ is $V$-periodic and that there exist $m \geq 1$ and
line polynomials $\varphi_1,\ldots,\varphi_m$ in pairwise non-parallel directions
$\bv_1,\ldots,\bv_m \in \Z^d$,
respectively,
such that $c$ is annihilated by the product $\varphi_1 \cdots \varphi_m$.
Also, assume that for all $i,j \in \{ 1 , \ldots , m \}, i \neq j$ it holds that
$$
\langle \bv_i,\bv_j \rangle \cap V = \{ \bo\}.
$$
% that $\bv_i$ is not in the subspace generated by $\bv_j$ and $V$, and
% $\bv_j$ is not in the subspace generated by $\bv_i$ and $V$.
Then there exist functions $c_1,\ldots,c_m \in \C^{\Z^d}$ such that
$$
c=c_1+\ldots+c_m
$$
where each $c_i$ is annihilated by $\varphi_i$ and
each $c_i$ is $V$-periodic.

Moreover, if the line polynomials $\varphi_1,\ldots,\varphi_m$ are difference polynomials and $c \in \Z^{\Z^d}$, then we may choose $c_1,\ldots,c_m \in \Z^{\Z^d}$.
% instead of $c_1,\ldots,c_m \in \C^{\Z^d}$.
\end{lemma}

\begin{proof}
We prove the claim by induction on $m$.
The case $m=1$ is clear by choosing $c_1=c$.
Let us assume then that $m >1$ and that the claim holds for $m-1$.
The function $\varphi_m c$
is annihilated by $\varphi_1\cdots \varphi_{m-1}$ and hence by the induction hypothesis there exist functions $c_1', \ldots , c_{m-1}' \in \C^{\Z^d}$ such that
$$
\varphi_m c = c_1' + \ldots + c_{m-1}'
$$
where each $c_i'$ is annihilated by $\varphi_i$ and $V$-periodic.
By Lemma \ref{lemma: auxiliary lemma 1 for main result 1} for each $i \in \{ 1, \ldots , m-1 \}$ there exists $c_i \in \C^{\Z^d}$ such that $\varphi_m c_i = c_i'$ and $\varphi_i c_i = 0$,
and $c_i$ is $V$-periodic.
Set $c_m = c-c_1 - \ldots - c_{m-1}$.
Then $c=c_1+\ldots+c_m$ and
\begin{align*}
\varphi_m c_m &= \varphi_m c - \varphi_m c_1 - \ldots - \varphi_m c _{m-1} \\
&= c_1'+\ldots+c_{m-1}' -c_1' - \ldots - c_{m-1}' \\ &=0.
\end{align*}
Finally, since $c$ and all $c_1, \ldots, c_{m-1}$ are $V$-periodic, also $c_m$ is $V$-periodic.

The ``moreover'' part follows from the ``moreover'' part of Lemma \ref{lemma: auxiliary lemma 1 for main result 1}.
% \qed
\end{proof}

\noindent
The following lemma is adapted from the proof of Theorem 3 in \cite{DLT_invited_jarkko}.

\begin{lemma}[\cite{DLT_invited_jarkko}]\label{lemma: auxiliary lemma 3 for main result 1}
% [Adapted from the proof of Theorem3 in \cite{DLT_invited_jarkko}]
    Let $V \in \R^d$ be a proper linear subspace of $\R^d$ ($\dim(V) \leq d-1$) and let $c \in \C^{\Z^d}$.
    Assume that $c$ has a periodizer $g$ such that $\supp(g) \cap V = \{\bo\}$.
    Then $c$ has also an annihilator $f$ such that $\supp(f) \cap V = \{\bo\}$.
\end{lemma}

Now, we are ready to prove our first improvement of the periodic decomposition theorem.

\begin{theorem}\label{thm: main result 1}
    Let $c$ be a $d$-dimensional configuration and let $k \in \{1, \ldots , d \}$.
    Assume that for every $(k-1)$-dimensional linear subspace $V \subseteq \R^d$ the configuration $c$ has a periodizer $f$ such that
    $\supp(f) \cap V = \{\bo\}$.
    Then there exist $k$-periodic functions $c_1,\ldots,c_m \in \Z^{\Z^d}$ such that
    $$
    c=c_1+\ldots+c_m.
    $$
\end{theorem}

\begin{proof}
We prove the claim by induction on $k$.
If $k=1$, then the assumption implies that the configuration $c$ has a non-trivial periodizer and hence by Theorem \ref{thm: decomposition theorem} it is a sum of periodic functions.

Let $k \in \{1,\ldots,d-1\}$ and assume
that the claim holds for $k$.
Let us prove that the claim holds then also for $k+1$.
So, assume that for all $V \in \G_{k}$ the configuration $c$ has a periodizer $f$ such that $\supp(f) \cap V = \{\bo\}$.
Since every $(k-1)$-dimensional subspace is contained in a $k$-dimensional subspace, there exists such periodizer, in particular, also for all $V \in \G_{k-1}$.
Thus, by the induction hypothesis there exist $k$-periodic functions $e_1,\ldots,e_l \in \Z^{\Z^d}$ such that
$$
c = e_1 + \ldots + e_l.
$$
For each $i\in\{1,\ldots,l\}$, let 
$V_i \in \G_k$ be such that $e_i$ is $V_i$-periodic.
We may assume that for all $i \neq j$ we have $V_i \neq V_j$ since the sum of two $V$-periodic functions is also $V$-periodic.
Indeed, if for some $i \neq j$ we have $V_i=V_j$, then we replace the functions $e_i$ and $e_j$ by the function $e_i+e_j$ which is also $V_i$-periodic.

% Consider $i=1$.
Let us fix an arbitrary $i \in \{1,\ldots,l\}$. By the assumption the configuration $c$ has a periodizer $f$ such that
$$
\supp(f) \cap V_i = \{\bo\}.
$$
By Lemma \ref{lemma: auxiliary lemma 3 for main result 1} we may assume that $f$ is an annihilator of $c$ and hence
by Theorem \ref{thm: special annihilator}
we conclude that $c$ has an annihilator
$$
(X^{\bu_1}-1) \cdots (X^{\bu_{r}}-1)
$$
such that $\bu_j \not \in V_i$ for all $j \in \{ 1, \ldots , r\}$ and the vectors $\bu_1,\ldots,\bu_r$ are pairwise non-parallel.

For each $j \in \{1,\ldots , l\}\setminus \{i\}$, let $\bw_j \in V_j$ be such that $\bw_j \not \in V_i$ and $e_j$ is annihilated by $X^{\bw_j}-1$.
Since $c=e_1+\ldots+e_l$ is annihilated by the product
$(X^{\bu_1}-1) \cdots (X^{\bu_r}-1)$
and each $e_j$ is annihilated by $X^{\bw_j}-1$, it follows that $e_i$ is annihilated by the polynomial
$$
(X^{\bu_1}-1) \cdots (X^{\bu_r}-1) \prod_{j \in \{1,\ldots , l\}\setminus \{i\}} (X^{\bw_j}-1).
$$

Thus, we have seen that $e_i$ is annihilated by a polynomial of the form
$$
(X^{\bv_1}-1) \cdots (X^{\bv_s}-1)
$$
where $\bv_j \not \in V_i$ for each $j \in \{1,\ldots,s\}$.
Let us assume that $s$ is minimal in the sense that if a product
$
(X^{\bv'_1}-1) \cdots (X^{\bv'_t}-1)
$
annihilates $e_i$ such that $\bv_j \not \in V_i$ for each $j \in \{1,\ldots,t\}$, then $t\geq s$.
We claim that the vectors $\bv_1,\ldots,\bv_s$ are pairwise non-parallel and $ \langle \bv_j, \bv_{j'} \rangle \cap V_i =  \{ \bo\}$ for all $j,j'\in \{1,\ldots,s\},j\neq j'$.

(Note that in the above annihilator we may replace any $\bv_{j}$ with $p \bv_{j}$ for any integer $p$ and the obtained polynomial is still an annihilator of $e_i$.
This is due to the fact that if a function is $\bv$-periodic, then it is also $p \bv$-periodic for any integer $p$.)

If the vectors $\bv_1,\ldots,\bv_s$ are not pairwise non-parallel,
then for some $j,j'\in \{1,\ldots,s\}, j \neq j'$ 
% the vectors $\bv_j$ and $\bv_{j'}$ are parallel and hence 
the product
$\varphi=(X^{\bv_j}-1)(X^{\bv_{j'}}-1)$ is a line polynomial in direction $\bv_j$.
Since $\varphi$ annihilates the function 
$$
\frac{(X^{\bv_1}-1) \cdots (X^{\bv_s}-1)}{(X^{\bv_j}-1)(X^{\bv_{j'}}-1)}
e_i,
$$
it is periodic in direction $\bv_j$, that is, annihilated by $X^{p\bv_j}-1$ for some non-zero $p \in \Z$.
Thus, $e_i$ is annihilated by 
$$
(X^{p\bv_j}-1) \cdot
\frac{(X^{\bv_1}-1) \cdots (X^{\bv_s}-1)}{(X^{\bv_j}-1)(X^{\bv_{j'}}-1)}
$$
which is a product of $s-1$ non-trivial difference polynomials.
This is a contradiction with the minimality of $s$ and hence the vectors $\bv_1,\ldots,\bv_s$ must be pairwise non-parallel.

Let us then show that $\langle\bv_j, \bv_{j'} \rangle \cap V_i =  \{ \bo\}$ for all $j,j'\in \{1,\ldots,s\},j\neq j'$.
Assume on the contrary that there exist $j,j' \in \{1,\ldots , s\},j \neq j'$ such that
$\langle\bv_j , \bv_{j'} \rangle \cap V_i \neq \{\bo\}$.
Then there exist integers $p, p'$ such that $p' \bv_{j'} =  p \bv_j + \bv$ for some $\bv \in V_i \setminus \{\bo\}$.
The function $e_i$ is periodic in direction $\bv$ and hence it is $k \bv$-periodic for some non-zero integer $k$.
We may assume that $k=1$.
If this is not the case, then
we replace $ \bv$ by $k  \bv$ and the integers $p$ and $p'$ are replaced by $kp$ and $kp'$, respectively.
Now, $e_i$ is $\bv$-periodic.
Since $\bv_j,\bv_{j'} \not \in V_i$ we have $p,p' \neq 0$.
Now, we replace the term $X^{\bv_{j'}}-1$ in the annihilator 
$(X^{\bv_1}-1) \cdots (X^{\bv_s}-1)$ of $e_i$ by $X^{p \bv_{j} + \bv}-1$.
Let us denote the obtained annihilator of $e_i$ by $g$.
Since $e_i$ is $\bv$-periodic, also the function
$$
e=\frac{g}{X^{p \bv_{j} + \bv}-1} e_i
$$ 
is $\bv$-periodic, that is, $X^{\bv}e=e$.
Since it is annihilated by $X^{p \bv_{j} + \bv}-1$, we have
$$
0=(X^{p \bv_{j} + \bv}-1)e = X^{p \bv_{j} + \bv} e -e = X^{p \bv_{j}}e - e= (X^{p \bv_{j}}-1)e
$$
and hence
it is also annihilated by $X^{p \bv_j}-1$.
Consequently, we replace the term $X^{p \bv_{j} + \bv}-1$ in $g$ by $X^{p \bv_j}-1$.
Finally, the term $(X^{p \bv_j}-1)(X^{\bv_j}-1)$ is replaced by $X^{q \bv_j}-1$ for suitable $q \in \Z$.
Again, we get a contradiction with the minimality of $s$ by obtaining an annihilator of $e_i$ which is a product of $s-1$ non-trivial difference polynomials.

Thus, by Lemma \ref{lemma: auxiliary lemma 2 for main result 1} there exist $V_i$-periodic functions $c_{1},\ldots,c_{s}$ such that for each $j \in \{1,\ldots,s\}$ the function $c_j$ is annihilated by $X^{\bv_j}-1$, that is, $\bv_j$-periodic and
$$
e_i = c_{1} + \ldots + c_{s}.
$$
So, each $c_j$ is $L (\{\bv_j \} \cup V_i)$-periodic and hence $(k+1)$-periodic since $\bv_j \not \in V_i$. 

This same reasoning works for any $i \in \{1,\ldots,l\}$. 
So, we conclude that each $e_i$ is a sum of $(k+1)$-periodic functions and hence $c=e_1 + \ldots + e_l$ is a sum of $(k+1)$-periodic functions.
% \qed
\end{proof}

% \noindent
Note that also the converse of the above theorem holds.
In other words, if
$c_1,\ldots,c_m \in \Z^{\Z^d}$ 
are $k$-periodic functions, then
their sum
$
c=c_1+\ldots+c_m,
$
has for all $V \in \G_{k-1}$ a periodizer $f$ such that $\supp(f) \cap V = \{\bo\}$.
Let us state this direction as a lemma and provide a proof:
% for the claim.

% \bigskip
% \noindent
% \textbf{Claim.}
% \emph{
\begin{lemma}
Let $k \in \{1,\ldots,d\}$, and
let $c_1,\ldots,c_m \in \Z^{\Z^d}$ be $k$-periodic functions.
Then for all $V \in \G_{k-1}$ the function $c=c_1+\ldots+c_m$ has a periodizer $f$ such that 
$$
\supp(f) \cap V = \{\bo\}.
$$
\end{lemma}
% }

\begin{proof}
We prove the claim by induction on $m$.
Assume first that $m=1$.
Let $V \in \G_{k-1}$.
Since $c=c_1$ is $k$-periodic, it has a period vector $\bv$ such that $\bv \not \in V$.
Now, $X^{\bv}-1$ is an annihilator (and hence also a periodizer) of $c$ satisfying $\supp(X^{\bv}-1) \cap V =  \{\bo\}.$

Assume then that $m \geq 2$ and that the claim holds for $m-1$.
Let $V \in \G_{k-1}$ be arbitrary.
By the induction hypothesis the function
$$
c'=c_2+\ldots+c_m
$$
has a periodizer $f'$ such that $\supp(f') \cap V = \{\bo\}$.
% By Lemma \ref{lemma: auxiliary lemma 3 for main result 1} we may assume that $f'$ is an annihilator of $c'$.
Again, let $\bv$ be such that $c_1$ is $\bv$-periodic and $\bv \not \in V$.
Note that $c_1$ is also $n \bv$-periodic and hence $X^{n \bv}-1$ is an annihilator of $c_1$ for all $n \in \Z$.
Thus, 
$
f=(X^{n \bv}-1)f'
$
is a periodizer of $c$ for all $n$.
For large enough $n$, we have $\supp(f) \cap V = \{\bo\}$.
% \qed
\end{proof}

\begin{remark}
A linear subspace $V \subseteq \R^d$ is \emph{expansive} for a subshift $\X \subseteq \A^{\Z^d}$ if there exists $r \in \Z_+$ such that for any $c,e \in \X$ the condition
$$
c(\bu) = e(\bu) \text{ for all } \bu \in V^r
$$
implies that $c=e$ where $V^r = \{ \bu \in \Z^d \mid d(\bu,\bv) \leq r \text{ for some } \bv \in V\}$ is the set of integer vectors within distance $r$ from $V$.
A theorem by Boyle and Lind states that if every $(d-1)$-dimensional subspace is expansive for a subshift $\X \subseteq \A^{\Z^d}$, then $\X$ is finite \cite{boyle-lind}.
If in Theorem \ref{thm: main result 1} we have $k=d$, then it is easily seen that every $(d-1)$-dimensional subspace is expansive for the orbit closure $\overline{\mathcal{O}(c)}$ of $c$ and hence $\overline{\mathcal{O}(c)}$ is finite by the theorem by Boyle and Lind. 
Thus, $c$ is strongly periodic in this setting by using a similar argument as in \cite{ballier}.
This case is also stated as Corollary 1 in \cite{DLT_invited_jarkko}.
For $k \in \{1,\ldots,d-1\}$, it may be that the functions in Theorem \ref{thm: main result 1} have infinitely many distinct coefficients, that is, they are not configurations as Remark \ref{remark: periodic decomposition theorem} suggests.
\end{remark}

\subsection*{Application to translational tilings}\label{subsection: translational tilings}

A \emph{translational tiling} is a binary configuration $c \in \{0,1\}^{\Z^d}$ such that for some non-empty finite set $D \subseteq \Z^d$ and a polynomial $f_D = \sum_{\bu \in -D} X^{\bu}$ we have
$$
f_Dc = \mathbbm{1}= 1^{\Z^d}.
$$
In this case $D$ is called a \emph{tile} and $c$ is a \emph{co-tiler of $D$} or a (translational) \emph{tiling by the tile $D$}.

Define that $k$ tiles $D_1,\ldots,D_k \subseteq \Z^d$ 
% satisfying 
% $\bo \in D_i$ for each $i \in \{1,\ldots,k\}$ 
% where $k \in \{1,\ldots,d\}$ 
are \emph{independent} if
$\bo \in D_i$ for each $i \in \{1,\ldots,k\}$
and the vectors $\bv_1,\ldots,\bv_k$ are linearly independent over $\Q$ for every tuple $(\bv_1,\ldots,\bv_k) \in D_1 \setminus \{\bo\} \times \ldots \times D_k \setminus \{\bo\}$ \cite{meyerovitch2023}.
It is known that if $k$ independent tiles have a common co-tiler $c$, then $c$ is a sum of finitely many $k$-periodic functions from the set $[a,b]^{\Z^d}$ for some reals $a<b$ \cite{meyerovitch2023}.
% \end{remark}
We will prove a similar result as a corollary of Theorem \ref{thm: main result 1}.
First, consider a simple example.

\begin{example}
Consider the two tiles $D_1=\{(0,0),(1,0)\},D_2=\{(0,0),(0,1)\} \subseteq \Z^2$.
They are clearly independent.
Moreover, they have a common co-tiler $c \in\{0,1\}^{\Z^2}$ defined such that
$c(i,j)=1$ if and only if $i+j$ is even.
\end{example}

We say that $k$ polynomials $f_1,\ldots,f_k$ satisfying $\bo \in \supp(f_i)$ for each $i \in \{1,\ldots,k\}$ are independent if their supports are independent.
% in the sense of \cite{meyerovitch2023}.

\begin{corollary}\label{corollary: independent periodizers}
Let $c \in \A^{\Z^d}$ be a configuration and let $f_1,\ldots,f_k$ be $k$ periodizers of $c$ satisfying $\bo \in \supp(f_i)$ for each $i \in \{1,\ldots,k\}$ with $k \in \{1,\ldots,d\}$.
If $f_1,\ldots,f_k$ are independent, then 
$$
c=c_1 + \ldots + c_m
$$
where each $c_i \in \Z^{\Z^d}$ is $k$-periodic.
\end{corollary}

\begin{proof}
Let us prove that for every $V \in \G_{k-1}$ the configuration $c$ has a periodizer $f$ such that $\supp(f) \cap V = \{\bo\}$.
Then the claim follows from Theorem \ref{thm: main result 1}.

Assume on the contrary that there exists a $V \in \G_{k-1}$ such that $\supp(f) \cap V  \not = \{\bo\}$ for any periodizer $f$ of $c$.
This implies that for any periodizer $f$ of $c$ we have either $\supp(f) \cap V = \emptyset$ or $|\supp(f) \cap V| \geq 2$.
Since $\bo \in \supp(f_i)$, we have $|\supp(f_i) \cap V| \geq 2$ for each $i \in \{1,\ldots,k\}$.
Thus, there exist non-zero vectors $\bv_1,\ldots,\bv_k$ such that $\bv_i \in \supp(f_i) \cap V$ for each $i \in \{1,\ldots,k\}$
and hence the vectors $\bv_1,\ldots,\bv_k$ span a linear subspace of dimension at most $k-1$.
This means that the sets $\supp(f_1),\ldots,\supp(f_k)$ are not independent and hence the polynomials $f_1,\ldots,f_k$ are not independent.
A contradiction.
% \qed
\end{proof}

We have the following immediate corollary.

\begin{corollary}
If a $d$-dimensional configuration $c \in \A^{\Z^d}$ has $d$ independent periodizers, then it is strongly periodic.
\end{corollary}

\begin{proof}
By the above corollary $c=c_1+\ldots+c_m$ where each $c_i$ is strongly periodic. Thus, $c$ is strongly periodic as a sum of finitely many strongly periodic configurations.
\end{proof}

% \noindent
For any tile $D\subseteq\Z^d$, the polynomial $f_D = \sum_{\bu \in -D}X^{\bu}$ is a periodizer of any co-tiler $c$ of $D$.
Thus, by Corollary \ref{corollary: independent periodizers}, any common co-tiler of $k$ independent tiles is a sum of finitely many $k$-periodic functions of the set $\Z^{\Z^d}$.
So, we have a similar result as the authors of \cite{meyerovitch2023} except that the $k$-periodic functions in the periodic decomposition of the co-tiler are now functions of the set $\Z^{\Z^d}$ instead of the set $[a,b]^{\Z^d}$.

\begin{corollary}\label{corollary: independent tiles}
Let $D_1,\ldots,D_k \subseteq \Z^d$ be $k$ independent tiles and let $c \in \{0,1\}^{\Z^d}$ be their common co-tiler.
Then there exist $k$-periodic functions $c_1,\ldots,c_m \in \Z^{\Z^d}$ such that
$$
c=c_1+\ldots+c_m.
$$
\end{corollary}

\begin{proof}
Follows directly from Corollary \ref{corollary: independent periodizers}.
\end{proof}

% \begin{lemma}

% \end{lemma}

\section{Periodic decomposition of sparse configurations} \label{sec: main result 2}

% \textcolor{blue}{Define sparse configurations..}
Let us denote by $C_m = \{ -m , \ldots , m\}^d$ for $m \in \N$ the discrete $d$-dimensional hypercube of size $(2m+1)^d$ centered at the origin. 
% In this paper 
Any function $c \in \C^{\Z^d}$ is called \emph{sparse} if there exists a positive integer $a$ such that
$$
|\supp(c) \cap (C_m + \bt)| \leq am
$$
for all $m \in \Z_+$ and for all $\bt \in \Z^d$.
In this case, $a$ is called a sparseness constant of $c$.

In this section we consider sparse configurations with annihilators and prove that they are sums of finitely many periodic fibers.
We prove the following theorem.

\begin{theorem} \label{thm: periodic decomposition sparse}
  Let $c$ be a sparse configuration and assume that it is annihilated by a product $\varphi_1 \cdots \varphi_n$ of line polynomials $\varphi_1,\ldots,\varphi_n$ in pairwise non-parallel directions $\bv_1,\ldots,\bv_n$, respectively.
  Then
  $$
  c= c_1 + \ldots + c_n
  $$
  where each $c_i$ is a sum of finitely many periodic $\bv_i$-fibers and $\varphi_i c_i =0$.
\end{theorem}

Before presenting the proof of the theorem
we state some simple observations concerning sparse configurations.

% \subsection{Properties of sparse configurations}

\begin{lemma}\label{lemma: a sum of sparse configurations is also sparse}
  Let $c_1,\ldots, c_n$ be sparse configurations and let $k_1,\ldots,k_n \in \Z$.
  Then also $k_1c_1+\ldots+k_nc_n$ is sparse.
  In particular, for a sparse configuration $c$ and a polynomial $g$ also $gc$ is sparse.
\end{lemma}

\begin{proof}
  There exist $a_1,\ldots,a_n$ such that
  $$
  |\supp(c_i) \cap (C_m + \bt)| \leq a_i m
  $$
  for all $m\in \Z_+$ and $\bt \in \Z^d$.
  Since 
  $$
  \supp(k_1c_1+\ldots+k_nc_n) \subseteq \cup_{i=1}^n \supp(c_i),
  $$
  we have
  $$
  |\supp(k_1c_1+\ldots+k_nc_n) \cap (C_m + \bt)| \leq (a_1+\ldots+a_n) m
  $$
  for all $m\in \Z_+$ and $\bt \in \Z^d$.
  Thus,
  $k_1c_1+\ldots+k_nc_n$ is sparse.

For a polynomial $g=\sum_{i=1}^n k_i X^{\bu_i}$ we have
$gc = k_1 c_1 + \ldots + k_n c_n$ where each $c_i = X^{\bu_i} c$ is sparse as a translation, \emph{i.e.}, a monomial multiplication of $c$.
  Thus, $gc$ is sparse by the first part.
% \qed
\end{proof}

\begin{lemma}\label{lemma: orbit closure of a sparse configuration}
Every element of the orbit closure $\overline{\mathcal{O}(c)}$ of a sparse configuration $c$ is also sparse.
\end{lemma}

\begin{proof}
Let $e \in \overline{\mathcal{O}(c)}$.
For all $m \in \Z_+$ any $C_m$-pattern of $e$ is also a $C_m$-pattern of $c$ and hence for any $\bt \in \Z^d$ there exists $\bt' \in \Z^d$ such that
$$
\supp(e) \cap (C_m + \bt) = \supp(c) \cap (C_m + \bt').
$$
Thus, if $c$ is sparse, then also $e$ is sparse.
% \qed
\end{proof}

\begin{remark}
To make the above lemma work it is important to define sparse configurations in the way we did.
For example, the weaker class of configurations $c$ having a constant $a$ such that
$$
|\supp(c) \cap C_m| \leq am
$$
for all $m \in \Z_+$ does not satisfy the above lemma.
Indeed, consider a configuration $c \in \{0,1\}^{\Z^2}$
such that $c(\bu)=1$ if $\bu \in C_m + (2^m,0)$ for some $m \in \Z_+$ and otherwise $c(\bu)=0$.
Now, $|\supp(c) \cap C_{m}| = O(m)$ and hence there exists $a$ such that
$$
|\supp(c) \cap C_{m}| \leq a m
$$
for all $m \in \Z_+$.
However, the constant configuration $1^{\Z^2}$ is in the orbit closure $\overline{\mathcal{O}(c)}$ and does not satisfy the condition.
\end{remark}

The following lemma is Theorem \ref{thm: periodic decomposition sparse} in the case $n=1$.

\begin{lemma} \label{lemma: a sparse periodic configuration is a sum of finitely many fibers}
  Let $c$ be a sparse configuration and assume that it is periodic in direction $\bv$ for some non-zero $\bv$.
  Then $c$ is a sum of finitely many periodic $\bv$-fibers.
\end{lemma}

\begin{proof}
Without loss of generality we may assume that $c$ is $\bv$-periodic by replacing $\bv$ by $m\bv$ for suitable $m \in \Z$.

  Assume on the contrary that $c$ is not a sum of finitely many periodic $\bv$-fibers, that is, it contains infinitely many distinct non-zero periodic $\bv$-fibers.
  Let $m \in \Z_+$ be such that $\bv \in C_m$,
  % Moreover, for any fiber $e$ let $n(e)$ be such that $\supp(e) \cap C_{n(e)} \neq \emptyset$.
  and let $e_1, \ldots , e_t$ be $t$ distinct non-zero fibers of $c$ where $t = am + 1$ and $a$ is a sparseness constant of $c$.
  Let $n \in \Z_+$ be such that $\supp(e_j) \cap C_n \neq \emptyset$ for each $j \in \{1,\ldots, t\}$.
  
  Since each $e_j$ is $\bv$-periodic, we have $|\supp(e_j) \cap C_{n + km}| > k$ for each $j \in \{1, \ldots , t \}$ and for all $k \in \Z_+$.
  Let $k = an + 1$.
  Then
  \begin{align*}
  |\supp(c) \cap C_{n+km}| & \geq \sum_{j=1}^t |\supp(e_j) \cap C_{n + km}| \\
  &> k t
  = (an+1)(am+1)=anam+an+am+1 \\
  &
  > an + aanm + am
  =a(n + km).
  \end{align*}
  This is a contradiction with $a$ being a sparseness constant of $c$.
  % \qed
\end{proof}

\noindent
In fact, also the converse of the above lemma holds, that is, if a configuration is a sum of finitely many periodic $\bv$-fibers, then it is sparse and periodic in direction $\bv$.
% , that is, a sum of finitely many periodic $\bv$-fibers is sparse and periodic.
Indeed, clearly any $\bv$-fiber $e$ is sparse since 
$$
\supp(e) \cap(C_m + \bt) \subseteq (\bu + \Q \bv) \cap(C_m + \bt)
$$
for some $\bu$
and $|(\bu+ \Q \bv) \cap(C_m + \bt) | \leq 2m+1$.
By Lemma \ref{lemma: a sum of sparse configurations is also sparse} a sum of finitely many sparse configurations is sparse and hence a sum of finitely many periodic $\bv$-fibers is sparse and periodic.
Thus, a configuration is sparse and periodic in direction $\bv$ if and only if it is a sum of finitely many periodic $\bv$-fibers.

The following lemma is Theorem \ref{thm: periodic decomposition sparse} in the case $n=2$.

\begin{lemma} \label{lemma: sparse n=2}
% THEOREM \ref{thm: periodic decomposition sparse} in the case $n=2$.
Let $c$ be a sparse configuration and assume that it is annihilated by the product $\varphi \psi$ of two line polynomials $\varphi$ and $\psi$ in non-parallel directions $\bv$ and $\bu$, respectively.
Then
$$
c= c_1 + c_2
$$
where $c_1$ is a sum of finitely many periodic $\bv$-fibers and $\varphi c_1 = 0$, and
$c_2$ is a sum of finitely many periodic $\bu$-fibers and $\psi c_2 = 0$.
\end{lemma}

\begin{proof}
Consider the configuration $e_1=\psi c$.
It is sparse by Lemma \ref{lemma: a sum of sparse configurations is also sparse}.
Moreover, $\varphi e_1 = \varphi \psi c =0$ and hence $e_1$ is periodic in direction $\bv$.
Thus, by Lemma \ref{lemma: a sparse periodic configuration is a sum of finitely many fibers} the configuration $e_1$ is a sum of finitely many periodic $\bv$-fibers.
Similarly, we conclude that the configuration $e_2 = \varphi c$ is a sum of finitely many periodic $\bu$-fibers.

Let $p\in \Z \setminus \{0\}$ be such that $e_1$ is $p \bv$-periodic.
Let $c_1$ be a limit of a converging subsequence of the sequence
% By compactness of the configuration space every sequence of configurations has a converging subsequence and hence there exist numbers $0<k_1<k_2<k_3<...$ such that
% the 
$$
c, X^{p\bv}c, X^{2p\bv}c, X^{3p\bv}c, \ldots
$$
of translated copies of $c$ by multiples of $p \bv$.
By compactness of the configuration space such subsequence exists.
Since $c_1 \in \overline{\mathcal{O}(c)}$ and $c$ is sparse, by Lemma \ref{lemma: orbit closure of a sparse configuration} also $c_1$ is sparse.
Let $k_1<k_2<k_3<\ldots$ be such that
$c_1 = \lim_{i \to \infty}X^{k_ip \bv}c$.
% Since the function $c \mapsto gc$ is a continuous in the configuration space for any polynomial $g$, 
We have
$$
\varphi c_1 
= \varphi \lim_{i \to \infty} X^{k_ip \bv}c=  \lim_{i \to \infty} X^{k_ip \bv} \varphi c = \lim_{i \to \infty} X^{k_ip \bv} e_2
=0.
$$
Above we used the fact that the function $e \mapsto g e$ is a continuous function in the topology for a configuration $e$ and polynomial $g$.
The final equality holds because $e_2$ is a sum of finitely many $\bu$-fibers that are moved by $k_i p \bv$-translations arbitrarily far away from the origin and hence disappear when we take the limit.
Moreover, we have
$$
\psi c_1 
=
\psi \lim_{i \to \infty} X^{k_ip \bv}c
=
\lim_{i \to \infty} X^{k_ip \bv} \psi c
=
\lim_{i \to \infty} X^{k_ip \bv} e_1
=
e_1 .
% = 
% \psi c.
$$
Again, we used the continuity of the function $e \mapsto g e$.
The final equality holds because $e_1$ is $p\bv$-periodic.
% The equality $$\lim_{i \to \infty} X^{k_ip \bv} e_1
% =
% e_1 $$ holds because $e_1$ is $p \bv$-periodic.

Similarly, we take $q \in \Z \setminus \{0\}$ such that $e_2$ is $q \bu$-periodic.
Then we define
$$
c_2 = \lim_{i \to  \infty} X^{t_i q \bu} c
$$
for a suitable sequence $t_1<t_2<t_3<\ldots$.
Again, since $c_2 \in \overline{\mathcal{O}(c)}$ and $c$ is sparse, by Lemma \ref{lemma: orbit closure of a sparse configuration} also $c_2$ is sparse.
By similar arguments as above, we have 
$$
\psi c_2 = \psi \lim_{i \to  \infty} X^{t_i q \bu} c
=
\lim_{i \to  \infty} X^{t_i q \bu} \psi c
=\lim_{i \to  \infty} X^{t_i q \bu} e_1 = 0
$$
and 
$$
\varphi c_2 
=
\varphi \lim_{i \to  \infty} X^{t_i q \bu} c
=
\lim_{i \to  \infty} X^{t_i q \bu} \varphi c
=
\lim_{i \to  \infty} X^{t_i q \bu} e_2
=
e_2.
% = \varphi c.
$$

Let us show that $c=c_1+c_2$.
% We do this by showing that $c-c_1-c_2=0$.
We have
$$
\varphi (c-c_1-c_2) = \varphi c -\varphi c_1 - \varphi c_2
= e_2 -0 -e_2=0.
$$
Hence, $c-c_1-c_2$ is periodic in direction $\bv$.
By Lemma \ref{lemma: a sum of sparse configurations is also sparse} it is sparse since $c$, $c_1$ and $c_2$ are all sparse.
It follows by Lemma \ref{lemma: a sparse periodic configuration is a sum of finitely many fibers} that $c-c_1-c_2$ is a sum of finitely many periodic $\bv$-fibers.
Similarly,
$$
\psi (c-c_1-c_2) = \psi c - \psi c_1 -\psi c_2 = e_1-e_1-0= 0
$$
and hence $c-c_1-c_2$ is a sum of finitely many periodic $\bu$-fibers.
So, $c-c_1-c_2$ is both a sum of finitely many periodic $\bv$-fibers and a sum of finitely many periodic $\bu$-fibers.
Since $\bv$ and $\bu$ are non-parallel, it follows that
$c-c_1-c_2=0$.

We have seen that 
$$
c=c_1+c_2
$$
where $\varphi c_1 = 0$ and $\psi c_2 =0$.
Moreover, both $c_1$ and $c_2$ are sparse.
Thus, $c_1$ is a sum of finitely many periodic $\bv$-fibers and $c_2$ is a sum of finitely many periodic $\bu$-fibers by Lemma \ref{lemma: a sparse periodic configuration is a sum of finitely many fibers}.
The claim follows.
% \qed
\end{proof}

Now, we are ready to prove Theorem \ref{thm: periodic decomposition sparse}.

\bigskip
\noindent
\emph{Proof of Theorem \ref{thm: periodic decomposition sparse}.}
The proof is by induction on $n$.
  First, let us consider the case $n=1$.
  So, $c$ is periodic in direction $\bv_1$.
  Then by Lemma \ref{lemma: a sparse periodic configuration is a sum of finitely many fibers} it is a sum of finitely many periodic $\bv_1$-fibers.

  Assume then that $n > 1$ and that the claim holds for $n-1$.
  Consider the configuration
  $$
  c' = \varphi_n c.
  $$
  It is sparse and annihilated by $\varphi_1 \cdots \varphi_{n-1}$.
  By the induction hypothesis
  $$
  c' = c_1' + \ldots + c_{n-1}'
  $$
  where each $c_i'$ is a sum of finitely many periodic $\bv_i$-fibers and annihilated by $\varphi_i$.

So, for any $i \in \{1,\ldots,n-1\}$ the configuration $c'_i$ is periodic in direction $\bv_i$, that is, $k_i \bv_i$-periodic for some $k_i \in \Z \setminus \{0\}$.
Consider the sequence
$$
c',X^{k_i \bv_i}c', X^{2k_i \bv_i}c', X^{3k_i \bv_i}c', \ldots.
$$
It converges to $c_i'$ and hence $c_i' \in \overline{\mathcal{O}(c')}$.
Let $e$ be a limit of a converging subsequence of the sequence 
$$
c,X^{k_i \bv_i}c, X^{2k_i \bv_i}c, X^{3k_i \bv_i}c, \ldots.
$$ 
By compactness of the configuration space such subsequence exists.

We have $\varphi_n e = c_i'$.
Since $c$ is sparse and $e \in \overline{\mathcal{O}(c)}$, by Lemma \ref{lemma: orbit closure of a sparse configuration} also $e$ is sparse.
Moreover, we have $\varphi_i \varphi_n e = 0$ since $\varphi_i c_i' = 0$.
By Lemma \ref{lemma: sparse n=2} we have
$$
e=e_i+e_n
$$
where $e_i$ is a sum of finitely many periodic $\bv_i$-fibers and $\varphi_i e_i=0$, and $e_n$ is a sum of finitely many $\bv_n$-fibers and $\varphi_n e_n=0$.
It follows that $\varphi_n e_i = \varphi_n e = c_i'$.

Now, we choose $c_i=e_i$ for $i \in \{1,\ldots,n-1\}$ and
$c_n = c - c_1 - \ldots - c_{n-1}$.
  Clearly, $c= c_1 + \ldots + c_n$. Moreover, we have
  \begin{align*}
  \varphi_n c_n &= \varphi_n c -\varphi_n c_1 - \ldots - \varphi_nc_{n-1}\\
  &=c' - c_1' -\ldots - c_{n-1}' \\ &= 0.
  \end{align*}
Since for each $i \in \{1,\ldots,n-1\}$ the configuration $c_i$ is sparse as a sum of finitely many periodic $\bv_i$-fibers and $c$ is sparse, also $c_n= c - c_1 - \ldots - c_{n-1}$ is sparse by Lemma \ref{lemma: a sum of sparse configurations is also sparse}.
% as a sum of finitely many sparse configurations.
  Thus, by Lemma \ref{lemma: a sparse periodic configuration is a sum of finitely many fibers} the configuration $c_n$ is a sum of finitely many periodic $\bv_n$-fibers.
\qed

\bigskip
% \noindent
Theorem \ref{thm: periodic decomposition sparse} together with Theorem \ref{thm: special annihilator} yields the following corollary.

\begin{corollary} \label{corollary: sparse configuration with annihilators is a sum of finitely many periodic fibers}
  Let $c$ be a sparse configuration and assume that it has a non-trivial annihilator.
  Then it is a sum of finitely many periodic fibers.
\end{corollary}

\begin{proof}
By Theorem \ref{thm: special annihilator} the configuration $c$ is annihilated by a polynomial
$$
(X^{\bv_1}-1) \cdots (X^{\bv_m}-1)
$$
where the vectors $\bv_1,\ldots,\bv_m$ are pairwise non-parallel and hence
by Theorem \ref{thm: periodic decomposition sparse} there exist configurations $c_1,\ldots,c_m$ such that each $c_i$ is a sum of finitely many periodic $\bv_i$-fibers and
$c=c_1+\ldots+c_m$.
Thus, $c$ is a sum of finitely many periodic fibers.
% \qed
\end{proof}

\subsection*{Application to $\R$-configurations with uniformly discrete supports that have annihilators}

In the following we consider one-dimensional $\R$-configurations, that is, functions $c \in \A^{\R}$ where $\A \subseteq \Z$ is a finite subset of integers.
Most of the concepts concerning $\Z^d$-configurations are naturally expanded for $\R$-configurations.
For example, the translation $\tau^v(c)$ of $c \in \A^{\R}$ by $v\in\R$ is defined such that $\tau^v(c)(u) = c(u-v)$ for all $u\in\R$.
Consequently, we say that $c$ is periodic if $\tau^v(c)=c$ for some non-zero $v$.
Naturally, the support $\supp(c) \subseteq \R$ of an $\R$-configuration $c$ is defined as the set 
% $\supp(c) \subseteq \R$ 
where $c$ gets non-zero values.

The concept of annihilators is also generalized to $\R$-configurations.
In other words, we say that an $\R$-configuration $c$ has a non-trivial annihilator if there exists a non-empty finite set $D=\{v_1,\ldots,v_n\} \subseteq \R$ and integers $a_1,\ldots,a_n \in \Z$ such that $a_1 \tau^{v_1}(c)+\ldots+a_n\tau^{v_n}(c)=0$.
We then say that $c$ is annihilated by the \emph{$\R$-polynomial} $f=f(x)=a_1x^{v_1}+\ldots+a_nx^{v_n}$ and may denote $fc=0$.

We say that a subset $S \subseteq \R$ is \emph{uniformly discrete} if
there exists a positive real number $r$ such that for any two distinct $s_1,s_2 \in S$ we have
$|s_1-s_2| \geq 2r$.
% such that any open ball of radius $r$ in $\R^d$ contains at most one point of $S$.
Largest such $r$ is called the \emph{packing radius} of $S$.

Next, we prove that any $\R$-configuration with a uniformly discrete support that has a non-trivial annihilator is necessarily periodic.
In the proof, we use the following lemma, which is based on the fact that every point has a dense orbit with respect to an irrational rotation of the circle \cite{kurka}.
We denote by $\lfloor x \rfloor$ and $\{x\}= x - \lfloor x \rfloor$ the integer and fractional parts, respectively, of a real number $x \in \R$.
% For the proof we need the following auxiliary lemma???

% \begin{lemma}
% APULEMMA??
% \end{lemma}
% EI EHKÄ SITTENKÄÄN LEMMAA!!??

\begin{lemma} \label{lemma: irrational rotations}
% linearly independent??*
Let $\alpha,\beta,r_1,r_2 \in \R$.
Moreover, assume that $r_1$ and $r_2$ are rationally independent, that is, $m_1 r_1=m_2 r_2$ if and only if $m_1=m_2=0$ where $m_1,m_2 \in \Z$.
For every $\varepsilon >0$, there exist $m_1,m_2 \in \Z$ such that $\alpha + m_1 r_1$ and $\beta + m_2 r_2$ are distinct and within distance $\varepsilon$ from each other.
\end{lemma}

\begin{proof}
By the assumption $\frac{r_1}{r_2}$ is an irrational number.
Consider the map $\rho \colon \R \to [0,1), \ x \mapsto \{x + \frac{r_1}{r_2}\}$.
It is well known that every point of $[0,1)$ has a dense orbit with respect to $\rho$ since $\frac{r_1}{r_2}$ is irrational \cite{kurka}.
This means that for every
$\varepsilon >0$ there exists $m$ such that
$$
0<\rho^m \left(\frac{\alpha-\beta}{r_2}\right ) < \frac{\varepsilon}{|r_2|}
$$
and hence
$$
% \left \{ \frac{\alpha-\beta}{r_2} + m_1 \cdot \frac{r_1}{r_2} \right \} =
0<\frac{\alpha-\beta}{r_2} + m_1 \cdot \frac{r_1}{r_2}
-\left \lfloor \frac{\alpha-\beta}{r_2} + m_1 \cdot \frac{r_1}{r_2} \right \rfloor
< \frac{\varepsilon}{|r_2|}.
$$
Now, by multiplying the inequality by $r_2$ and denoting $m_2 = \left \lfloor \frac{\alpha-\beta}{r_2} + m_1 \cdot \frac{r_1}{r_2} \right \rfloor$ we have
$$
0<|\alpha + m_1r_1 - (\beta + m_2r_2)|< \varepsilon. 
$$
Thus, $\alpha + m_1r_1$ and $\beta + m_2r_2$ are distinct and within distance $\varepsilon$.
\end{proof}

\begin{theorem}\label{thm: delone application}
  Let $c \in \A^{\R}$ be an $\R$-configuration with a uniformly discrete support.
  If $c$ has a non-trivial annihilator,
  then $c$ is periodic.
\end{theorem}

\begin{proof}
  Let $f = f(x)=a_1x^{t_1}+\ldots+a_sx^{t_s}$ be a non-trivial annihilator of $c$ and let $d$ be the rank of the additive abelian group $\Z[\supp(f)]$.
  Let $\{b_1,\ldots,b_d\}$ be a (minimal) generator set of this group.
  Let us define a projection map
  $$
  \pi \colon \R^d \to \R, \ (i_1,\ldots,i_d) \mapsto i_1b_1+\ldots+i_db_d.
  $$
  By the fundamental theorem of finitely generated abelian groups \cite{abstract-algebra} the group $\Z[\supp(f)]$ is isomorphic to $\Z^d$ since it has no torsion elements as a subgroup of $\R$.
  This means that every element of the group $\Z[\supp(f)]$ has a unique presentation of the form $i_1b_1+\ldots+i_db_d$ where $i_1,\ldots,i_d\in \Z$, that is,
  the restriction $\pi \restriction _ {\Z^d}$ is injective.
  
  % Note that every element of $\Z[\supp(f)]$ has a unique presentation in the form $i_1b_1+\ldots+i_db_d$ where $i_1,\ldots,i_d\in \Z$, that is, the function $(i_1,\ldots,i_d) \mapsto i_1b_1+\ldots+i_db_d$ is injective.
  % This follows from the fundamental theorem of finitely generated abelian groups \cite{abstract-algebra}.
  
  Define for all $\alpha \in \R$, a $\Z^d$-configuration $c^{(\alpha)} \in \A^{\Z^d}$ such that
  $$
  c^{(\alpha)}(i_1,\ldots,i_d) = c(\alpha + i_1 b_1 + \ldots + i_d b_d).
  $$
  Let us show that $c^{(\alpha)}$ is sparse for all $\alpha$.
  So, consider an arbitrary $c^{(\alpha)}$.
  Since $\supp(c)$ is uniformly discrete, there exists 
  %a rational number 
  $\delta >0$ such that each half-open interval
  $$I_k = [\alpha + k \delta, \alpha + (k+1) \delta)$$ contains at most one point of $\supp(c)$ for each $k\in\Z$.
  % (Moreover, the intervals $I_k$ where $k\in\Z$ partition the real line $\R$.)
  Thus, the sets
  $$
  S_k
  % = (c^{(\alpha)})^{-1}(I_k)
  = \{ (i_1,\ldots,i_d) \in \R^d \mid k \delta \leq i_1b_1 + \ldots + i_db_d < (k+1) \delta \}
  $$
  contain at most one element of $\supp(c^{(\alpha)}) \subseteq \Z^d$ by the mentioned injectivity of the map $\pi \restriction_{\Z^d}$. 
  % function $(i_1,\ldots,i_d) \mapsto i_1b_1+\ldots + i_db_d$.
  %Note that the set $S_k$ is an intersection of a closed and an open half-space.
  For $m\in\Z_+$ and $\bt \in \Z^d$, let $N(m,\bt)$ denote the number of sets $S_k$ that intersect the set $C_{m} +\bt$.
  % We show that there exists $a$ such that
  % $N(m,\bt) \leq a m$ for all $m \in \Z_+$ and $\bt \in \Z^d$.
  % TÄMÄ PITÄÄ KIRJOITTAA PAREMMIN!!
  Note that the set $C_m + \bt$ intersects $S_k$ if and only if the set $\pi(C_m + \bt) + \alpha$ intersects $I_k$.
  The set $\pi(C_m + \bt) + \alpha$ is an interval of length $(2m+1)l$ where $l$ is the length of the interval $\pi([0,1]^d)$ --- the image of the unit cube under the projection map $\pi$.
  So, we conclude that the set $\pi(C_m + \bt) + \alpha$ intersects with at most $(2m+1)l/\delta + 2$ sets $I_k$.
  In particular, there
  exists $a \in \Z_+$ such that
  $N(m,\bt) \leq a m$ for all $m \in \Z_+$ and $\bt \in \Z^d$.
  % and hence we can take $a= \lfloor l/\delta \rfloor + 2$.
  % It grows linearly with respect to $m$, that is, 
  % There exists $a$ such that
  % $N(m,\bt) \leq a m$ for all $m \in \Z_+$ and $\bt \in \Z^d$.
  % PERUSTELU!!
  % Indeed,.....
  Thus,
  $$
  |\supp(c^{(\alpha)}) \cap (C_m+\bt)| \leq N(m,\bt) \leq a m
  $$
  for all $m\in \Z_+$ and $\bt \in \Z^d$ and hence $c^{(\alpha)}$ is sparse.

% By the assumption $c$ has a non-trivial annihilator.
% and hence the $\Z^d$-configurations $c^{(\alpha)}$ have a common non-trivial annihilator.
Let us show that
the $\Z^d$-configurations $c^{(\alpha)}$ have a common non-trivial annihilator.
% By the assumption $c$ has a non-trivial annihilator,
% say 
% $$
% f(x)=a_1x^{t_1}+\ldots+a_sx^{t_s}
% $$
% where $a_1,\ldots,a_s\in \R\setminus\{0\},t_1,\ldots,t_s \in \Z[\supp(f)]$ and $s \geq 2$.
% be a non-trivial annihilator of $c$.
Now, let $\bt_i \in \Z^d$ be such that $\pi(\bt_i)=t_i$ for each $i$.
Denote $f_{\pi}(X) = a_1 X^{\bt_1}+ \ldots + a_s X^{\bt_s}$.
For any $\alpha\in\R$ and $\bu \in \Z^d$, we have
\begin{align*}
    f_{\pi}c^{(\alpha)}(\bu) &= a_1 c^{(\alpha)}(\bu-\bt_1) + \ldots + a_s c^{(\alpha)}(\bu-\bt_s) \\
    &= a_1 c(\alpha + \pi(\bu)-t_1) + \ldots + a_s c(\alpha + \pi(\bu)-t_s) \\
    &= (fc)(\alpha+\pi(\bu)) \\
    &= 0
\end{align*}
since $f$ is an annihilator of $c$.
Thus, $f_{\pi}$ is a common annihilator of the $\Z^d$-configurations $c^{(\alpha)}$.
% TÄMÄN VOISI PERUSTELLA!!
% PROJEKTIOT KÄYTTÖÖN....    

Now, by Theorem \ref{theorem: special annihilator common} there exist pairwise linearly independent vectors $\bv_1,\ldots,\bv_n \in \Z^d$ such that every $c^{(\alpha)}$ is annihilated by the same polynomial
$$
(X^{\bv_1}-1) \cdots (X^{\bv_n}-1).
$$
Thus, by Theorem \ref{thm: periodic decomposition sparse} we have
$$
c^{(\alpha)} = c^{(\alpha)}_1 + \ldots + c^{(\alpha)}_n
$$
where each $c^{(\alpha)}_i$ is a sum of finitely many $\bv_i$-periodic $\bv_i$-fibers.
Denote $\bv_i=(v_{i,1},\ldots,v_{i,d})$ for each $i \in \{1,\ldots,n\}$.
% YKSI YHTEINEN TULOS JOTA KÄYTETÄÄN SITTEN MOLEMMISSA KOHDISSA!!
Let us show that for all $\alpha$, we have $c^{(\alpha)} = c_i^{(\alpha)}$ for the same $i \in \{1,\ldots,n\}$.
% For this we will show that what exactly??

First, we will show that
for all $\alpha$, we have
$c^{(\alpha)} = c_i^{(\alpha)}$ for some $\alpha$ and $i \in \{1,\ldots,n\}$.
Assume on the contrary that
$c_i^{(\alpha)} \neq 0$ and $c_j^{(\alpha)} \neq 0$ for some $i \neq j$.
Then
there exist $\bu_i=(u_{i,1},\ldots,u_{i,d})$ and $\bu_j=(u_{j,1},\ldots,u_{j,d})$ such that
$c^{(\alpha)}_i(\bu_i) \neq 0$ and $c^{(\alpha)}_j(\bu_j) \neq 0$.
Since $c_i^{(\alpha)}$ is $\bv_i$-periodic and $c_j^{(\alpha)}$ is $\bv_j$-periodic, it follows that
$c^{(\alpha)}_i(\bu_i + t \bv_i) \neq 0$ and $c^{(\alpha)}_j(\bu_j + t \bv_j) \neq 0$ for all $t \in \Z$.
Thus, $c(\alpha_i + t(v_{i,1}b_1 + \ldots + v_{i,d}b_d))\neq 0$ and $c(\alpha_j + t(v_{j,1}b_1 + \ldots + v_{j,d}b_d))\neq 0$ for all $t \in \Z$ where
$\alpha_i=\alpha + u_{i,1} + \ldots +u_{i,d}$ 
and 
$\alpha_j=\alpha + u_{j,1} + \ldots +u_{j,d}$.
If 
$$
m(v_{i,1}b_1 + \ldots + v_{i,d}b_d) = m'(v_{j,1}b_1 + \ldots + v_{j,d}b_d)
$$
for some $m,m'\in \Z$,
then $mv_{i,1}=m'v_{j,1},\ldots, mv_{i,d}=m'v_{j,d}$ 
by the injectivity of the map $\pi \restriction _{\Z^d}$.
Thus, $m\bv_i =m'\bv_j$ and hence
$m=m'=0$ since $\bv_i$ and $\bv_j$ are linearly independent.
This means that the numbers $ v_{i,1}b_1+\ldots+v_{i,d}b_d$ and $ v_{j,1}b_1+\ldots+v_{j,d}b_d$ are rationally independent.
% Without loss of generality, we may assume that
Consequently,
for all $\varepsilon >0$, there exist $m,m' \in \Z$ such that $\alpha_i + m(v_{i,1}b_1 + \ldots + v_{i,d}b_d)$ and $\alpha_j + m'(v_{j,1}b_1 + \ldots + v_{j,d}b_d)$ are distinct and within distance $\varepsilon$ from each other by Lemma \ref{lemma: irrational rotations}.
This is a contradiction with the uniform discreteness of $\supp(c)$.

Finally, let us show that for all $\alpha$ we have $c^{(\alpha)}=c_i^{(\alpha)}$ for the same $i \in \{1,\ldots,n\}$.
Assume on the contary that we have for
some $\alpha$ and $\beta$
that $c^{(\alpha)}=c_i^{(\alpha)}$ and $c^{(\beta)} = c_j^{(\beta)}$ for $i \neq j$.
% Samanlainen perustelu kuin yllä ja saadaan ristiriita.
Let
$\bu_i=(u_{i,1},\ldots,u_{i,d})$ and $\bu_j=(u_{j,1},\ldots,u_{j,d})$ be such that
$c^{(\alpha)}_i(\bu_i) \neq 0$ and $c^{(\beta)}_j(\bu_j) \neq 0$ and
denote $\alpha_i = \alpha + u_{i,1} + \ldots +u_{i,d}$ and $\alpha_j = \beta + u_{j,1} + \ldots +u_{j,d}$.
Again, we have 
$c(\alpha_i + t(v_{i,1}b_1 + \ldots + v_{i,d}b_d))\neq 0$ and $c(\alpha_j + t(v_{j,1}b_1 + \ldots + v_{j,d}b_d))\neq 0$ for all $t \in \Z$.
We noticed already above
that
the numbers $ v_{i,1}b_1+\ldots+v_{i,d}b_d$ and $ v_{j,1}b_1+\ldots+v_{j,d}b_d$ are rationally independent.
% , we have by Lemma \ref{lemma: irrational rotations}
Thus, again by Lemma \ref{lemma: irrational rotations} there exist $m,m' \in \Z$ such that
$\alpha_i + m(v_{i,1}b_1 + \ldots + v_{i,d}b_d)$ and $\alpha_j + m'(v_{j,1}b_1 + \ldots + v_{j,d}b_d)$
are different but arbitrarily close.
A contradiction with the uniform discreteness of $\supp(c)$.

So, we conclude that for all $\alpha \in \R$, we have $c^{(\alpha)} = c_i^{(\alpha)}$ for the same $i\in\{1,\ldots,n\}$.
It follows that every $c^{(\alpha)}$ is $\bv_i$-periodic and hence  
$c$ is $(v_{i,1} b_1 + \ldots + v_{i,d}b_d)$-periodic.
\end{proof}

Finally, we give an example that shows that the above theorem does not hold for general $\R$-configurations that have annihilators.

\begin{example}
% Let $\alpha \in \R$ be an irrational number.
Consider the $\R$-configurations
$c_1,c_2 \in \{0,1\}^{\R}$ defined such that
  $$
  c_1(i) = 
  \begin{cases}
    1 &\text{, if } i \in \Z \\
    0 &\text{, otherwise}
  \end{cases}
  $$
  and
  $$
  c_2(i) = 
  \begin{cases}
    1 &\text{, if } i \in \alpha \Z  \\
    0 &\text{, otherwise}
  \end{cases}
  $$
  where $\alpha$ is an irrational number.
  % $c_1(x) = \sum_{i \in \Z} x^i$ and $c_2(x) = \sum_{i \in \Z} x^{i \alpha}$ where $\alpha$ is an irrational number. 
  Both $c_1$ and $c_2$ are periodic, that is, annihilated by non-trivial difference $\R$-polynomials. 
  More precisely, $c_1$ is annihilated by $x - 1$ and $c_2$ is annihilated by $x^{\alpha} - 1$. 
  Consequently, their sum $c = c_1 + c_2 \in \{0,1,2\}^{\R}$ has a non-trivial annihilator 
  $(x - 1)(x^{\alpha} - 1)$. However, $c$ is non-periodic since $\alpha$ is irrational.
\end{example}

% \section{Conclusion}

% We proved to improvements of the periodic decomposition theorem and gave for each improvement an application.
% The first improvement

\bibliographystyle{plain}
\bibliography{biblio}

\newpage

\end{document}